\pdfoutput=1
\documentclass[11pt]{article}

\usepackage{microtype,amsmath,amssymb,amsthm, latexsym,graphicx, xspace, url, color, cite, subfigure}
\graphicspath{{Figures/}}

\usepackage{hyperref}

\usepackage[medium,compact]{titlesec}
\usepackage{enumerate}

\usepackage{euscript}

\usepackage{fullpage}

\newtheorem{theorem}{Theorem}[section]
\newtheorem{lemma}[theorem]{Lemma}

\let\eps\varepsilon

\def\A{\EuScript{A}}

\def\D{\EuScript{D}}
\def\E{\EuScript{E}}
\def\F{\EuScript{F}}

\def\P{\EuScript{P}}

\def\S{\EuScript{S}}

\def\V{V}

\def\etal{\textsl{et~al.}}

\newcommand{\Term}[1]{\textsc{#1}}
\newcommand{\PTAS}{\Term{PTAS}\xspace}

\def\Opt{\textsc{Opt}}

\long\def\remove#1{}
\def\reals{\mathbb{R}}

\def\HH{\EuScript{H}}

\def\Vor{\EuScript{V}}

\makeatletter
\long\def\@makecaption#1#2{
   \vskip 10pt
   \setbox\@tempboxa\hbox{{\footnotesize \textbf{#1.} #2}}
   \ifdim \wd\@tempboxa >\hsize         %
       {\footnotesize \textbf{#1.} #2\par}%
     \else                              %
       \hbox to\hsize{\hfil\box\@tempboxa\hfil}
   \fi}
\makeatother

\usepackage{graphicx}

\begin{document}

\title{On Pseudo-disk Hypergraphs%
  \thanks{%
    Work on this paper by Boris Aronov has been supported
    by NSA MSP Grant H98230-10-1-0210, 
    by NSF Grants CCF-08-30691, CCF-11-17336, CCF-12-18791, and CCF-15-40656,
    and by BSF grant 2014/170.
    Work on this paper by Anirudh Donakonda has been partially 
    supported by NSF Grant CCF-11-17336.
    Work on this paper by Esther Ezra has been supported by NSF under grants  CAREER CCF-15-53354, CCF-11-17336, and CCF-12-16689.
    Work on this paper by Rom Pinchasi has been supported by ISF grant No.~409/16.
  }
}

\author{
  Boris Aronov\thanks{%
     Department of Computer Science and Engineering,
     Tandon School of Engineering, New York University,
     Brooklyn, NY~11201, USA;
     \textsl{boris.aronov@nyu.edu}, \textsl{ad2930@nyu.edu}.
   }
   \and
     Anirudh Donakonda\footnotemark[2]
   \and
   Esther Ezra\thanks{%
     Department of Computer Science, Bar-Ilan University, Ramat Gan, Israel and
     School of Mathematics, Georgia Institute of Technology, Atlanta, Georgia 30332, USA;
     \textsl{eezra3@math.gatech.edu};
     work on this paper by Esther Ezra was
     initiated when she was at the Courant Institute of Mathematical
     Sciences and then at the Department of Computer Science and Engineering,
     Polytechnic (now Tandon) School of Engineering,
     New York University.
   }
   \and
   Rom Pinchasi\thanks{%
     Department of Mathematics,
     Technion -- Israel Institute of Technology,
     Haifa, Israel 32000;
     \textsl{room@math.technion.ac.il}.
   }
}
\maketitle

\begin{abstract}
Let $\F$ be a family of pseudo-disks in the plane, and $\P$
be a finite subset of $\F$. Consider the hypergraph $H(\P,\F)$ whose vertices 
are the pseudo-disks in $\P$ and the edges are all subsets of~$\P$ of the
form $\{D \in \P \mid D \cap S \neq \emptyset\}$, where $S$ is a pseudo-disk in $\F$.
We give an upper bound of~$O(nk^3)$ for the number of edges in~$H(\P,\F)$
of cardinality at most~$k$. This generalizes a result of Buzaglo \etal~(2013). %

As an application of our bound, we obtain an algorithm 
that computes a constant-factor approximation  to the smallest \emph{weighted} dominating set in a collection of pseudo-disks in the plane, in expected polynomial time.
\end{abstract}

\section{Introduction}
\label{sec:introduction}

For a family of pseudo-disks $\F$ and a subset $\P \subset \F$, we denote by
$H(\P,\F)$ the hypergraph whose vertex set is $\P$ and whose
edges are all subsets of $\P$ of the form 
$\{D \in \P \mid D \cap S \neq \emptyset\}$, with~$S$ a~pseudo-disk from~$\F$.  That is, such a subset consists of all pseudo-disks in~$\P$ intersected by a fixed pseudo-disk of~$\F$.

Our main goal in this paper is to obtain an upper bound on the number of edges in~$H(\P,\F)$ of bounded cardinality. %
Specifically, we establish the following main property:

\begin{theorem}
  \label{theorem:main_c}
  Suppose $\F$ is a family of pseudo-disks in the plane and $\P$ is a finite subset of $\F$.
  Let $k \ge 1$ be an integer parameter.
  Then the number of edges of cardinality at most $k$ in $H(\P,\F)$ is $O(|\P|k^3)$, where the implied constant does not depend on the family $\F$.
\end{theorem}

Our proof technique exploits several ideas from the work of Buzaglo~\etal \cite{BPR-13},
who studied the corresponding problem for points and pseudo-disks in the plane.
Specifically, Buzaglo~\etal \cite{BPR-13} studied hypergraphs defined by \emph{points} and \emph{pseudo-disks enclosing them},
whereas we consider a hypergraph of \emph{pseudo-disks} and subsets of them 
\emph{intersected by pseudo-disks}.  Our result is a generalization of that in \cite{BPR-13},
as one can represent a point by a sufficiently small pseudo-disk.

Note that it is crucial to consider pseudo-disks rather than pseudo-circles
(that is, entire regions rather than just their boundaries).
Indeed, range spaces of pseudo-\emph{circles} and subsets of them met
by a pseudo-circle do not satisfy Theorem~\ref{theorem:main_c}:
Consider $n$ pairwise intersecting circles in general position and for each of the $n \choose 2$ pairs of circles, place a tiny circle at one of their intersection points.
Obviously, this construction yields a collection of quadratically many circle pairs,
contradicting the linear bound asserted in Theorem~\ref{theorem:main_c}, for $k=2$.

As an application of Theorem~\ref{theorem:main_c}, combined with the machinery of Chan~\etal~\cite{CGKS-12},
we show that the \emph{dominating set} of smallest weight in a collection of pseudo-disks in the plane can be approximated
up to a constant factor in expected polynomial time; to the best of our knowledge, the result for the weighted version of this problem was previously unknown. The details are presented in Section~\ref{sec:dominating_set}.

\section{Proof of Theorem~\ref{theorem:main_c} }
\label{sec:main_result}

\subsection{Preliminaries}
\label{sec:prelim}

\paragraph*{Family of pseudo-disks.}
A family of pseudo-disks is a set of objects in the plane, where each object is bounded by a Jordan curve and any two object boundaries either are disjoint, cross properly exactly twice, or are tangent exactly once. No boundary overlaps are allowed. Several boundaries may meet at a common point.
\paragraph*{Arrangements and levels.}
Let $\P \subset \F$ be a finite family of pseudo-disks in the plane.
Let $\A(\P)$ denote the \emph{arrangement} of $\P$ (see, e.g., \cite{AS-survey}).
The \emph{level} of an (open) face in this arrangement is the number of pseudo-disks containing it in their interior.
Well-known results by Kedem \etal \cite{KLPS-86} and by Clarkson and Shor\cite{CS-89} imply that 
$\A(\P)$ has $O(|\P|)$ level-2 faces and, more generally, $O(|\P|k)$ faces at level at most $k$.

\paragraph*{VC-dimension.}
Given a hypergraph $H$ with vertex set $X$,
we say that a subset $K \subseteq X$ is \emph{shattered} by $H$ if, 
for every subset $Z$ of $K$, 
$Z = K \cap e$ for some edge $e \in H$.  
The \emph{VC-dimension} of $H$ is the size of the largest
finite shattered subset.

The rest of this section is organized as follows.
We first show that, for a family $\F$ of pseudo-disks and a finite subset $\P \subset \F$, the 
VC-dimension of $H(\P,\F)$ is at most four and that this bound is optimal (Theorem~\ref{thm:vc_dim}).
Then we prove that the number of edges of
$H(\P,\F)$ of cardinality at most $k$ is linear in $|\P|$; the proof gives a super-polynomial dependency on $k$. 
Finally, using the above bound on the VC-dimension and the proof
technique in \cite{BPR-13}, we are able to improve the dependency on $k$ and show that 
the number of edges in $H(\P,\F)$ of cardinality at most $k$ is $O(|\P|k^3)$. 

\bigskip

\subsection{The Analysis}

We first show:

\begin{theorem}
  \label{thm:vc_dim}
  A hypergraph $H(\P,\F)$ as defined above has VC-dimension at most four. This bound is the best possible.
\end{theorem}

We start by stating the following technical lemma from~\cite{BPR-13}
(see figure below): %
\begin{lemma}[Buzaglo~\etal \cite{BPR-13}]
  \label{lemma:even_crossings}
  Let $\gamma$ and $\gamma'$ be arbitrary non-overlapping curves contained in
  pseudo-disks $D$ and~$D'$, respectively.  If the endpoints of
  $\gamma$ lie outside of $D'$ and the endpoints of $\gamma'$ lie
  outside of $D$, then $\gamma$ and $\gamma'$ cross an even
  number of times (where tangency counts as two crossings).
\end{lemma}
\begin{figure}[h]
  \centering
  \includegraphics{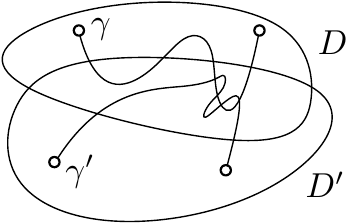}
\end{figure}

We say that a set~$K$ of pseudo-disks is \emph{well-behaved} if every pseudo-disk in~$K$ has a point not covered 
by the union of other pseudo-disks in $K$.

We begin with an auxiliary construction.  Let $K$ be a finite well-behaved set of pseudo-disks.
We construct a graph $G = G(K)$ whose vertices correspond to 
the pseudo-disks in $K$ and whose edges correspond to pseudo-disks in $\F$ 
that meet precisely two sets in $K$.  
More specifically, we draw $G$ as follows:

\paragraph{Vertices of $G$:}
For each pseudo-disk $D \in K$, we fix a point $v(D)\in D$ (which need not lie on the boundary of $D$), not contained in any other pseudo-disk of $K$; it exists since $K$ is well behaved.
The points $\{v(D) \mid D \in K\}$ form the \emph{vertex set} of $G$.

\paragraph{Edges of $G$:}
Let $D_1, D_2 \in K$, $v_1$ = $v(D_1)$ and $v_2$ = $v(D_2)$. Suppose there exists
$S \in \F$ that intersects $D_1$ and $D_2$ and no other disk in $K$; fix one such $S$ (it is possible that $S \in K$). We will add an \emph{edge} $v_1v_2$ to $G$, drawn as described below. We call a connected portion of the edge contained in~$S$ a \emph{red} arc and such a portion outside $S$ a \emph{blue} arc. The edge $v_1v_2$ consists of at most one red arc and at most two blue arcs.
In the figures below, we use the convention of drawing pseudo-disks of $K$ in blue and the ``connecting'' pseudo-disk(s) from $\F$ in red.
\begin{description}
\item[\normalfont\emph{$S$ contains both $v_1$ and $v_2$}:]
Draw a red arc in~$S$ from $v_1$ to $v_2$. 
This forms the edge $v_1v_2 \in G$.
See figure below.
\begin{figure}[h]
  \centering
  \includegraphics{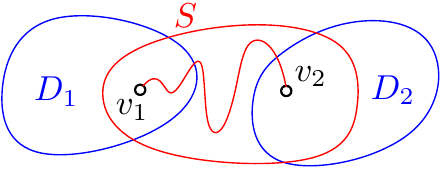}\hspace*{2ex plus  0.5fil}\includegraphics{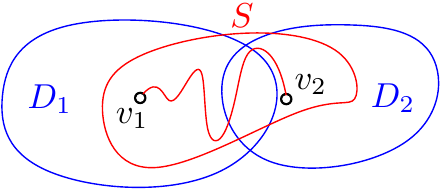}
\end{figure}
\item[\normalfont\emph{$S$ contains $v_1$, but not $v_{2}$}:]
  Draw a red arc in $S$ that starts at $v_1$ and ends at the boundary of $S$ 
  in~$D_2$. Now draw a blue arc in $D_2$ that starts at this point, ends at $v_2$ and lies completely outside $S$ otherwise. The concatenation of these two arcs forms the edge $v_1v_2$ of $G$.
  See figure below.
\begin{figure}[h]
  \centering
  \includegraphics{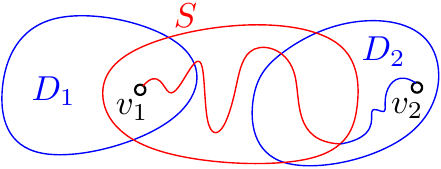}%
  \hspace*{1ex plus  0.5fil}\includegraphics{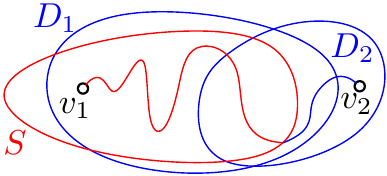}%
  \hspace*{1ex plus  0.5fil}\includegraphics{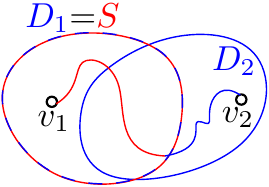}
\end{figure}
\item[\normalfont\emph{$S$ contains neither $v_1$ nor $v_2$}:]
Draw a blue arc in $D_1$ that starts at $v_1$, ends at the boundary of $S$ in~$D_1$, and otherwise stays outside of $S$. From its endpoint, draw a red arc in~$S$ to a point of the boundary of $S$ in~$D_2$ and from there, draw the final blue arc outside $S$ in $D_2$ to the vertex $v_2$. The concatenation of these three arcs constitutes the edge $v_1v_2$.
See figure below.
\begin{figure}[h]
  \centering
  \includegraphics{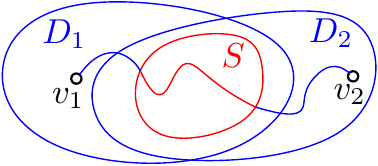}\hspace*{2ex plus  0.5fil}\includegraphics{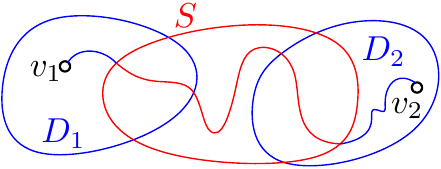}
\end{figure}
\end{description}
The simple but important observation that makes the construction above
possible is that, if $A$ and $B$ are two pseudo-disks, then both 
$A \setminus B$ and $B \setminus A$ are arcwise connected.
By construction, for each arc of the constructed edge, either red or blue, there is a pseudo-disk that completely contains~it. We also assume that the arcs belonging to different edges of $G$ may intersect at a finite number of points, but do not overlap among themselves.  Similarly, we will assume they do not overlap the boundaries of the finite number of pseudo-disks under consideration. 

\begin{lemma} 
  \label{lemma:planar}
  The graph $G = G(K)$ is planar.
\end{lemma}

\begin{proof}
We will prove $G$ is planar using the strong Hanani-Tutte theorem \cite{Tutte-Com-70}. Consider two edges~$e, e'$ that connect $v_1= v(D_1)$ to $v_2 = v(D_2)$, and $v_3 = v(D_3)$ to $v_4 = v(D_4)$ in $G(K)$, respectively, and do not share a vertex so that $D_1, D_2, D_3, D_4 \in K$ are pairwise distinct.  We will prove that $e$ and $e'$ intersect an even
number of times, by considering their red and blue portions
separately. Let $S \in \F$ be the pseudo-disk intersecting $D_1$ and $D_2$ 
and no other disk in $K$ that was used to draw $e$, and let 
$S' \in \F$ be the corresponding pseudo-disk intersecting only 
$D_3$ and $D_4$ from the disks in $K$. 

\textsc{Red-Blue Intersections:} Consider the red portion of $e$. This red arc is contained in $S$ and therefore does not meet any pseudo-disk of $K$ other than~$D_1, D_2$. As the blue portions of $e'$ lie inside $D_3, D_4$, this implies that the red arc of $e$ does not meet the blue portions of $e'$. Symmetrically the red portion of $e'$ cannot intersect the blue portions of $e$.

\textsc{Red-Red Intersections:} The red arc $\alpha$ along $e$ lies entirely in $S$ and has one endpoint in~$D_1$ and the other in $D_2$. Similarly, the red arc $\alpha'$ along $e'$ lies entirely in $S'$ and has one endpoint in~$D_3$ and the other in $D_4$. As $S$ does not intersect $D_3$~and~$D_4$ and $S'$ does not intersect $D_1$~and~$D_2$, the endpoints of~$\alpha$ do not lie in $S'$ and the endpoints of~$\alpha'$ do not lie in $S$.
By Lemma~\ref{lemma:even_crossings}, $\alpha$ and~$\alpha'$ intersect an even number of times.

\textsc{Blue-Blue Intersections:}  Consider blue arcs $\beta \subset e$  and $\beta' \subset e'$. The blue arc $\beta$ starts, say, at vertex $v_1$ of pseudo-disk $D_1$ and ends at $x$ in $D_1$ on the boundary of pseudo-disk $S$, and $\beta'$ starts, say, at vertex $v_3$ of pseudo-disk $D_3$ and ends at $x'$ in $D_3$ on the boundary of pseudo-disk $S'$.
By the construction of the vertices of $G$, we have 
$v_1 \notin D_3$ and $v_3 \notin D_1$. Now, $x$ cannot lie in $D_3$ because $S$ meets only $D_1$ and $D_2$ and similarly $x'$ cannot lie in $D_1$. Hence, by Lemma~\ref{lemma:even_crossings}  we deduce once again that $\beta$ and $\beta'$ intersect an even number of times.

There is a possibility that some edges of $G$ self-intersect, but such intersections can be removed using standard methods: see, for example, \cite{PSS07} and Figure~\ref{fig:self_intersect}.
\begin{figure}[h]
  \centering
  \includegraphics[width=0.9\textwidth]{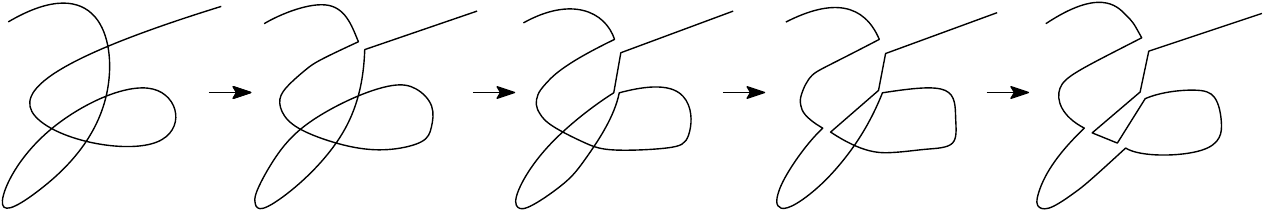}
  \caption{How to undo self-intersections}
  \label{fig:self_intersect}
\end{figure}
Thus, any two edges of $G$ that do not share an endpoint cross an even number of times, and therefore $G$ is planar by the strong Hanani-Tutte theorem \cite{Tutte-Com-70}.  
\end{proof}

\begin{proof}[Proof of Theorem~\ref{thm:vc_dim}]
  Let $K \subseteq \P$ be a set shattered by $\F$.
Since $K$ is shattered, for every pseudo-disk $P \in K$ there
is a pseudo-disk $F \in \F$ that intersects $P$ and no other element of~$K$.
Therefore, $K$ is well-behaved.

For a well-behaved set $K$, $G(K)$ is planar, by Lemma~\ref{lemma:planar},
and therefore has at most $3|K| - 6$~edges (if $|K|<3$, we are already done). However, $K$ is shattered, so $G(K)$ is a complete graph with ${|K| \choose 2}$~edges. Therefore, ${|K| \choose 2} \leq 3|K| - 6$, implying $|K| \le 4$.

This proves that the VC-dimension of $H(\P,\F)$ is at most four.
Figure~\ref{fig:four_shattered} shows that this bound is the best possible,
\begin{figure}[h]
  \centering
  \includegraphics[scale=0.75]{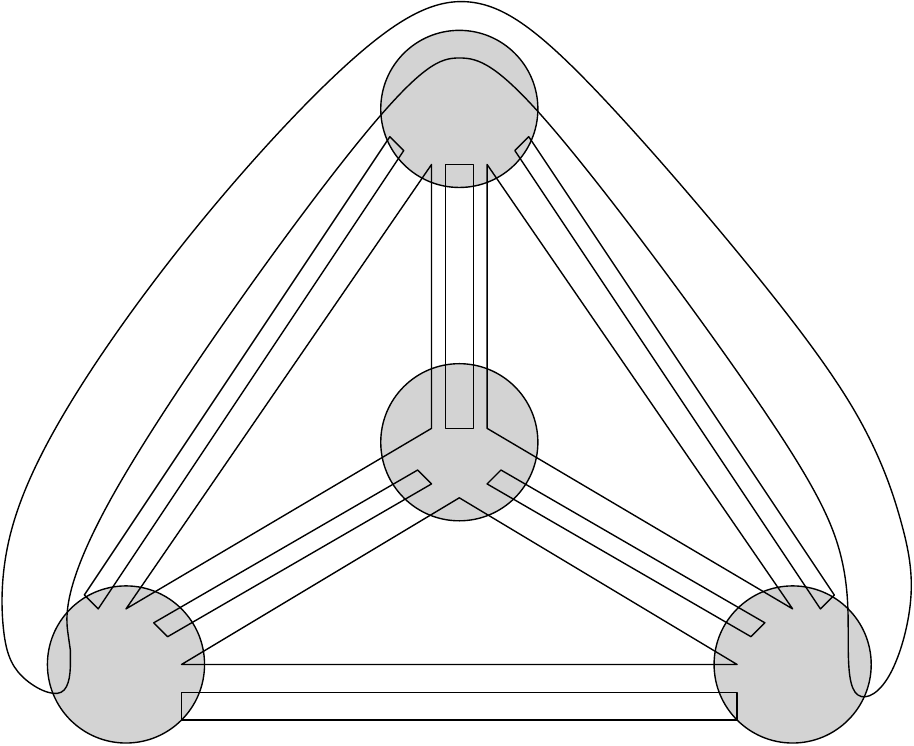}
  \caption{How to shatter a set of four pseudo-disk objects (shaded): the
    pseudo-disks meeting all or none of the four objects are not shown.
    The pseudo-disks meeting exactly one object are the objects
    themselves.}
  \label{fig:four_shattered}
\end{figure}
 completing the proof of Theorem~\ref{thm:vc_dim}.
\end{proof}

Using the analysis above, %
we first show that the number of edges of cardinality two in $H(\P,\F)$ is linear in $|\P|$:

\begin{theorem}
  \label{thm:pairs}
  Let $\F$ be a family of pseudo-disks and let $\P \subset \F$. %
  Then the number of edges of cardinality two in $H(\P,\F)$ is $O(|\P|)$.
\end{theorem}

\begin{proof}
First, consider the subset $K$ of $\P$ consisting of pseudo-disks $D$ with the property that $D$ contains a point $v(D)$ not covered by any other pseudo-disk of $\P$.  $K$ is well-behaved, by construction, and consequently, by Lemma~\ref{lemma:planar}, the set of edges of cardinality two it induces in~$H(\P,\F)$ forms a planar graph, and therefore its cardinality must be $O(|\P|)$.

It remains to consider edges $e$ in $H(\P,\F)$ of the form $\{D_1,D_2\}$ with $D_1$ covered by other pseudo-disks of $\P$, without loss of generality. By definition of $H(\P,\F)$, there must exist a pseudo-disk $S \in \F$ that meets $D_1$, $D_2$, and no other pseudo-disk of $\P$.  Notice that $D_1$ and $D_2$ must intersect, for otherwise, as $D_1$ is completely covered by other pseudo-disks of $\P$, it is impossible that $S$ intersects $D_1$, $D_2$, and no other pseudo-disks of $\P$ ($S$ would have to intersect one of the pseudo-disks covering~$D_1$, in addition to $D_1$ and $D_2$).  

Since $D_1$ is completely covered by other pseudo-disks of $\P$, there must exist a point $p$ of $D_1\cap D_2$ contained in $S$ and no other pseudo-disks of $\P$.

If $D_{1} \cap D_{2}$ contains an (open) face~$f$ of level~two in the arrangement $\A(\P)$, then we charge the edge~$e$ to~$f$ (then $p$ can be chosen to lie in~$f$). At most one edge is charged to~$f$.
Recalling that the number of faces of level two in $\A(\P)$ is $O(|\P|)$, we conclude that the number of such edges $e$ is $O(|\P|)$.

Now suppose that $e$ is not charged to any face of level two in $\A(\P)$.  Then 
the point $p$ chosen above must lie on the boundary of $D_1\cap D_2$ and not be contained in any other pseudo-disk of $\P$.  In particular,  in $\A(\P)$, it must either (a)~coincide with a vertex of level~zero ($p$ lies on the boundary of both $D_1$ and $D_2$ and is not contained in any other pseudo-disk of $\P$) or (b)~lie in an (open) edge of level~one ($p$ must be either contained in the interior of $D_1$ and on the boundary of $D_2$, or vice versa).

Now consider a neighborhood~$U$ of $p$ sufficiently small to avoid all other pseudo-disks of $\P$.  
In case~(b), it is easy to check that within~$U$ the edge of $\A(\P)$ containing $p$ would have to bound a level-two face contained in $D_1\cap D_2$, a situation that we have already excluded above.  In case~(a), examining all possibilities (the boundaries of~$D_1$ and of~$D_2$ may properly cross or touch at~$p$; $D_1$~and~$D_2$ may touch externally or internally), $U$ must meet a level-two face contained in~$D_1\cap D_2$ (excluded above) or a level-one face contained in~$D_1$ (also excluded, as we assumed $D_1$ is fully covered by other pseudo-disks of $\P$), or both.  Therefore neither case~(a) nor~(b) arises, 
thereby concluding the proof of the theorem.
\end{proof} 

We next show:
\begin{theorem}
  \label{thm:triples_quads}
  Let $\F$ be a family of pseudo-disks, let $\P \subset \F$ be a subset of 
  $\F$, and let $k \geq 2$ be a~fixed integer. 
  Then the number of edges in $H(\P,\F)$ of cardinality at most $k$ is $O_k(|\P|)$.
\end{theorem}

In order to prove Theorem~\ref{thm:triples_quads}, we first need the following key lemma:
\begin{lemma}
  \label{lemma:pairs}
  Let $k \geq 2$ be a fixed integer. Let $\F$ be a family of pseudo-disks in
  the plane. Let $\HH$ be a subfamily of $m$ pseudo-disks from $\F$. 
  We call a pair of pseudo-disks $\{D_{1},D_{2}\}$ from $\HH$ \emph{$k$-good} 
  if there exists a pseudo-disk in $\F$ that intersects $D_{1}$, $D_{2}$, and at most $k-2$ additional pseudo-disks from $\HH$,
  for a total of at most $k$ pseudo-disks from $\HH$.
  Then the number of $k$-good pairs in~$\HH$ is at most $c_{k}m$, where $c_{k}$ is an absolute constant depending only on~$k$.
\end{lemma}

\begin{proof}
We prove the lemma by induction on $k$. Case $k=2$ is precisely Theorem~\ref{thm:pairs}. 
Suppose $k \geq 3$. We choose each pseudo-disk in $\HH$ independently with probability $p = 1/2$ (but we keep $\F$ intact).
We denote the resulting sample of pseudo-disks by $\HH'$.
We say that a $k$-good pair $\{D_1, D_2\}$ from $\HH$ \emph{survives} if $D_1, D_2 \in \HH'$ 
and there is a pseudo-disk in $\F$ that intersects $D_{1}, D_{2}$, and a total of at most $k-1$ pseudo-disks in $\HH'$.
In other words, after sampling $\{D_{1},D_{2}\}$ becomes \emph{$(k-1)$-good}.
We observe that a $k$-good pair $\{D_{1},D_{2}\}$ in $\HH$ survives with probability of at least $1/8$.
Indeed, because $\{D_{1}, D_{2}\}$ is a $k$-good pair, there exists $F \in \F$ such that $F$ intersects $D_{1}$ and $D_{2}$
and a total of $\ell \leq k$ pseudo-disks in $\HH$. If $\ell \leq k-1$, then $\{D_{1}, D_{2}\}$
is $(k-1)$-good as soon as both $D_{1}$ and $D_{2}$ are in~$\HH'$; this happens with probability $1/4$.
If $\ell=k$, let $S \in \HH$ be a pseudo-disk other than $D_{1}$ and $D_{2}$ intersected by $F$. If $D_{1}$ and~$D_{2}$ are in 
$\HH'$ and $S$ is not in $\HH'$, then $\{D_{1},D_{2}\}$ becomes $(k-1)$-good.
This happens with probability $1/8$; there may be other ways for $\{D_1,D_2\}$ to become $(k-1)$-good.
Therefore the expected number of $(k-1)$-good pairs in $\HH'$ is at least 
$\frac18$ of the number of $k$-good pairs in $\HH$.

By the inductive hypothesis on $\HH'$, %
there are at most $c_{k-1}|\HH'|$ $(k-1)$-good pairs of pseudo-disks in $\HH'$.
Therefore, the expected number of $(k-1)$-good pairs of pseudo-disks in $\HH'$
is at most $c_{k-1} m/2$.

Combining the two estimates, the number of $k$-good pairs in $\HH$ is at most $4 c_{k-1} m$, as claimed.
\end{proof}

Theorem~\ref{thm:triples_quads} is then proved using the following result from \cite{BPR-13}:

\begin{lemma}[Buzaglo~\etal \cite{BPR-13}]
  \label{lem:l_cycles}
  Consider a graph $G$ on $m$ vertices, with the property that, in any subgraph induced by a subset $V$ of vertices, 
  the number of edges is at most $c |V|$, where $c > 0$ is an absolute constant. Then, for any $k \ge 2$, the number of copies of $K_k$
  (the complete graph on~$k$~vertices) in $G$ is at most $d_k m$, 
  where $d_k = \frac{(2c)^{k-1}}{k!}$.
\end{lemma}

\begin{proof}[Proof of Theorem~\ref{thm:triples_quads}]
We follow the approach in~\cite{BPR-13}.
We define a graph $G$ whose vertex set is~$\P$. Two pseudo-disks in $\P$
form an edge in $G$ if they belong to an edge of $H(\P,\F)$ 
of cardinality $k$.

By Lemma~\ref{lemma:pairs}, if $G'$ is any induced subgraph of $G$, then the number of edges in $G'$ is $O(|V(G')|)$,
where $V(G')$ is the set of vertices of $G'$.

We can now use Lemma~\ref{lem:l_cycles} and conclude that the number of 
copies of $K_{k}$, the complete graph on $k$ vertices,
in $G$ is $O(|V(G)|)$. This is sufficient to prove the assertion in Theorem \ref{thm:triples_quads},
since every edge of cardinality $k$ in $H(\P,\F)$ 
gives rise to a unique copy of $K_{k}$ in $G$, as is easily verified.
\end{proof}

Next we would like to prove Theorem \ref{theorem:main_c},
namely, to show that the number of edges
of cardinality at most $k$ in $H(\P,\F)$ is $O(|\P|k^3)$.
The bound in Theorem~\ref{thm:triples_quads} is linear in $|\P|$
but at the cost of a~multiplicative constant that grows extremely fast (super-exponentially) in~$k$.
In order to overcome this problem and improve the dependence on~$k$, we use Theorem~\ref{thm:triples_quads} and
a fundamental property shown in~\cite{BPR-13}, namely, that in a set system of bounded VC-dimension
every set has a unique small signature. Specifically:

\begin{theorem}[Buzaglo~\etal \cite{BPR-13}]
  \label{thm:signature}
  Let $\S = \{S_1, \ldots, S_m\}$ be a set family
  with VC-dimension $d$. Then it is possible to assign to each set
  $S \in \S$ a subset $S^{*} \subseteq S$ (its \emph{signature}), of cardinality at~most~$d$, 
  so that distinct sets from~$\S$ are assigned distinct signatures.
\end{theorem}
Given this machinery we are ready to prove Theorem \ref{theorem:main_c}.
We follow almost verbatim the random sampling approach in~\cite{BPR-13}. 
By Theorem~\ref{thm:vc_dim}, the VC-dimension of $H(\P,\F)$ 
is at most four. Applying Theorem~\ref{thm:signature},
we assign to each $e \in H(\P,\F)$ a unique subset $B_{e} \subseteq e$ of 
cardinality at most four.
  
  Let $0 < q < 1/2$ be a parameter to be fixed shortly. 
  We now select each pseudo-disk in~$\P$ independently with probability $q$.
  Let $\P'$ be the resulting sample, and consider the induced hypergraph $H(\P',\F)$.
  We say that $e \in H(\P,\F)$ \emph{survives} if all the pseudo-disks in $B_{e}$ 
  are in $\P'$ but none of the remaining pseudo-disks in $e \setminus B_{e}$ are in $\P'$.
  
  It is easy to verify that, if $e$ has cardinality
  at most $k$, then $e$ survives with probability
  $$
  q^{|B_{e}|} (1-q)^{|e| - |B_{e}|} \ge  q^{|B_{e}|} (1-q)^{k - |B_{e}|} \ge q^{4} (1-q)^{k-4} ,
  $$
  where the first inequality follows from the assumption $|e| \le k$, and the second from the fact that $q < 1/2$.
  
  By Theorems~\ref{thm:pairs} and~\ref{thm:triples_quads}, 
  the number of edges in $H(\P',\F)$ of cardinality two, three, and four is $O(|\P'|)$, with an absolute constant of proportionality.
  Clearly, the number of edges in $H(\P',\F)$ 
  of cardinality one is at most $|\P'|$. 
  It thus follows that the number of surviving edges from $H(\P,\F)$ is~$O(|\P'|)$, by Theorem~\ref{thm:signature}.
  Taking expectations, we see that the expected number of surviving edges 
from $H(\P,\F)$ is $O(|\P'|)=O(q|\P|)$.
  On the other hand, the expected number of surviving edges of $H(\P,\F)$ 
with cardinality at most $k$ is at least
  $q^{4} (1-q)^{k-4}Z$, where $Z$ is the number of edges in $H(\P,\F)$
of cardinality at most $k$.
  Therefore, $q^{4} (1-q)^{k-4}Z = O(q|\P|)$. 
  By setting $q = 1/k$, we obtain $Z = O(|\P| k^{3})$, as asserted.

  This at last completes the proof of Theorem~\ref{theorem:main_c}.

\section{An application to the weighted dominating set problem}
\label{sec:dominating_set}

\paragraph*{Problem statement.}
We are given a finite collection $\P$ of \emph{pseudo-disks} in the plane. 
We define the \emph{intersection graph} $G$ of $\P$ in the standard manner, that is,
the vertex set is $\P$ and there is an edge between two pseudo-disks if their intersection is non-empty.

The \textsc{dominating set} problem for $G$ is to find a smallest subset $\D \subseteq \P$, such that 
each vertex in $G$ is either in $\D$ or is adjacent to a vertex in $\D$. In other words, this is a smallest
subset of $\P$ such that any pseudo-disk in $\P$ appears in the subset or is intersected by a pseudo-disk
in it.
In the \textsc{weighted dominating set} problem, each element of $\P$ is assigned a non-negative weight,
and the goal is to find a dominating set of smallest total weight.

\paragraph*{Related work.}

It is beyond the scope of this paper to report all previous studies related to the dominating set problem.
We only mention that the abstract problem for general graphs is 
NP-hard to solve \cite{GJ-79, Karp-72}, and that the standard greedy algorithm yields an
$(1+\ln{n})$-approximation factor \cite{Chvatal-79, LOVASZ-75}, 
where $n$~is the size of the vertex set. 
The problem remains NP-hard
in more specialized settings, such as unit disk graphs and growth-bounded graphs \cite{CCJ-90}.
However, the approximation factors achievable in polynomial time tend to be better.
Specifically, the \textsc{dominating set} problem admits a polynomial-time approximation scheme~(\PTAS)
for the aforementioned settings \cite{HMRRRS-98, NH-06}; see also \cite{EM-09} for a constant-factor
approximation for the \textsc{weighted dominating set} problem on unit disk graphs.
The current state-of-the-art for pseudo-disk graphs is a \PTAS for the unweighted case, which recently has been
introduced by Govindarajan~\etal~\cite{GRRR-16}.
See also the earlier studies by Erlebach and van Leeuwen \cite{EL-08} for special forms of triangles and for axis-parallel
rectangles, and by Gibson and Pirwani~\cite{GP-10}, who obtained a \PTAS for the case of disk graphs, and a constant-factor
approximation for the weighted problem.
The latter result was strengthened by Chan~\etal~\cite{CGKS-12}, who also presented a simple
reduction from \textsc{set cover} to \textsc{dominating set}, considerably simplifying the approach taken in \cite{GP-10}.
For a more detailed discussion, we refer the reader to~\cite{GP-10, GRRR-16} and the references therein.

In this section we deduce the following main result, using the assertions in Theorem~\ref{theorem:main_c}, combined with 
the recent machinery of Chan~\etal \cite{CGKS-12}: 

\begin{theorem}
  \label{thm:dominating_set}
  There is a randomized expected polynomial-time algorithm, that, given a set~$\P$ of~pseudo-disks in the plane,
  each with a non-negative weight, 
  computes a dominating set~$\D \subseteq \P$ of~weight~$O(\Opt)$, 
  where $\Opt$ is the smallest total weight of such a dominating set.
\end{theorem}

We first show the connection between the dominating set problem and the \emph{hitting-set} problem, and then describe the machinery
of Chan~\etal \cite{CGKS-12} and how to apply it in the scenario of our problem.

\paragraph*{Hitting sets and dominating sets.}

Fix any family $\P$ of pseudo-disks in the plane.
Consider the intersection graph $G$ of $\P$ as defined above.
In $G$, a \emph{neighborhood} of a pseudo-disk is the set of pseudo-disks intersecting
it; therefore, this is a subgraph of $G$ spanned by (the vertex set of) a star. Note that we include the pseudo-disk itself in its neighborhood.
The family of all neighborhoods defines a hypergraph $H(\P)$, which is a special case of the hypergraph $H(\P, \F)$ defined above,
as in this case we have $\F = \P$.

We now observe that a dominating set in $G$ is, in fact, a \emph{hitting set} for $H(\P)$, where the latter
refers to a subset $\D \subseteq \P$, which meets all edges of $H(\P)$. 
That is, $\D$ meets all objects in~$\P$ if and only if each neighborhood in the
intersection graph (that is, an edge of $H(\P)$) is hit by an element of~$\D$. 
In particular, the \textsc{minimum hitting set} for $(\P,\E)$ corresponds to the \textsc{minimum dominating set} of $G$,
and this property holds in the weighted setting as well.
See Figure~\ref{fig:example} for an example. 

\begin{figure} %
  \centering  
    \subfigure[]{\includegraphics{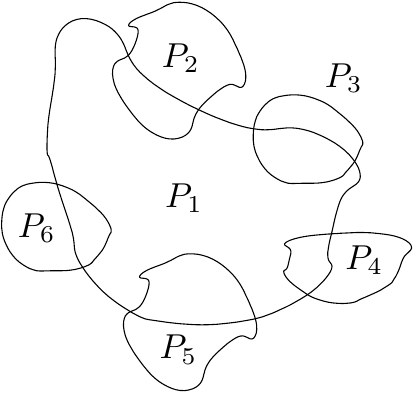}\label{fig:disks}}%
    \qquad
    \subfigure[]{\includegraphics{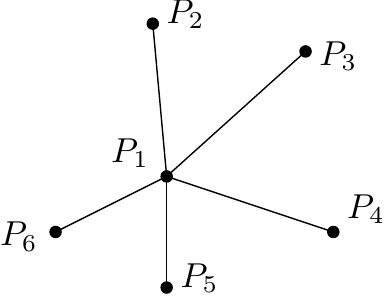}\label{fig:graph}}
    \caption{The intersection graph in \subref{fig:graph} induced by the pseudo-disks $P_1, \ldots, P_6$, depicted in \subref{fig:disks}, 
    is a ``star'' centered at $P_1$. The smallest dominating set is $\{ P_1 \}$. 
    The underlying hypergraph is $H(\P)$, where $\P = \{P_1, \ldots, P_6\}$, and the edges are
    $\{ \{P_1, \ldots, P_6\}, \{P_1, P_2\} , \{P_1, P_3\}, \{P_1, P_4\}, \{P_1,  P_5\}, \{P_1, P_6\} \}$. 
    Finally, $\{ P_1 \}$ is the smallest hitting set for this hypergraph.}
  \label{fig:example}
\end{figure}

Chan~\etal \cite{CGKS-12} showed the existence of small approximation factors (achievable in expected polynomial time)
for the weighted hitting-set problem in favorable scenarios. Specifically, they showed:

\begin{theorem}[Chan~\etal \cite{CGKS-12}]
  \label{thm:chan_weighted_cover}
  Let $H(V,E)$ be a hypergraph representing a hitting set instance, where the number of edges of cardinality $k$
  for any restriction of $H$ to a subset $V' \subseteq V$ is at most $O(|V'| k^{c})$, where $c > 0$ is an absolute
  constant and $k \le |V'|$ is an integer parameter.\footnote{In~\cite{CGKS-12} this property is referred to as ``shallow cell complexity,'' although we do not define it formally in this paper.}
  Then there exists a randomized polynomial-time $O(1)$-approximation algorithm for the weighted hitting set problem for $H(V,E)$.
\end{theorem}

It is now easy to verify that the statement in Theorem~\ref{thm:dominating_set} is obtained by combining 
Theorem~\ref{thm:chan_weighted_cover}, our observation regarding the equivalence between hitting sets and dominating
sets (in the sense discussed above), as well as our main result established in Theorem~\ref{theorem:main_c}.

%

\section*{Discussion}
\label{sec:discussion}
An earlier version of this work was presented at the 2015 Fall Workshop on Computational Geometry in Buffalo, NY (\url{https://www.cse.buffalo.edu/fwcg2015/assets/pdf/FWCG_2015_paper_15.pdf}).  Bal\'azs Keszegh has recently pointed out to us that a construction essentially identical to the one in Lemma~\ref{lemma:planar} has appeared independently in \cite{balazs:arXiv}.  He also noted that, just as in \cite{balazs:arXiv}, Theorem~\ref{theorem:main_c} and, by extension,  Theorem~\ref{thm:dominating_set}, also apply to the following, more general setting, unmodified: We once again consider the intersection hypergraph~$H(\P,\F)$, but allow $\P$~and~$\F$ to be two completely unrelated families of pseudo-disks.
$\P$ forms the ground set as above, and the hypergraph edges are formed by subsets of pseudo-disks from $\P$ intersected by a pseudo-disk from $\F$. We believe our analysis extends to this case as well, although we have not verified it in full detail.

\end{document}